\documentclass[11pt,a4paper,notitlepage]{article}

\usepackage[a4paper,margin=2.5cm]{geometry}
\usepackage[utf8]{inputenc}

\usepackage{microtype}
\usepackage{todonotes}
\usepackage{mathtools}

\usepackage{amsthm}\usepackage{amsmath} \usepackage{amsfonts} \usepackage{amssymb}  
\usepackage{hyperref}
\usepackage[numbers,sort]{natbib}
\title{An Optimal Construction for the Barthelmann-Schwentick Normal Form on Classes~of~Structures~of~Bounded Degree \\
	{\small Preliminary Version}}

\author{André Frochaux, Lucas Heimberg \\
Humboldt-Universität zu Berlin,\\ 
  \texttt{\{andre.frochaux,lucas.heimberg\}@informatik.hu-berlin.de} }

       \newtheorem{theorem}{Theorem}
        \newtheorem{lemma}[theorem]{Lemma}

        \theoremstyle{definition}

           \newtheorem{claim}{Claim}
        \newtheorem*{remark}{Remark}

\newcommand{\logic}[1]{\textup{\sf #1}}
\newcommand{\FO}{\logic{FO}}
\newcommand{\HNF}{\textsc{hnf}}
\newcommand{\GNF}{\textsc{gnf}}
\newcommand{\BSNF}{\textsc{bsnf}}
\newcommand{\ov}[1]{\overline{#1}}
\newcommand*{\und}{\ensuremath{\wedge}}
\newcommand*{\Und}{\ensuremath{\bigwedge}}
\newcommand*{\oder}{\ensuremath{\vee}}
\newcommand*{\Oder}{\ensuremath{\bigvee}}

\newcommand*{\impl}{\ensuremath{\rightarrow}}
\newcommand*{\gdw}{\ensuremath{\leftrightarrow}}
\newcommand{\NN}{\ensuremath{\mathbb{N}}}
\newcommand{\Npos}{\ensuremath{\mathbb{N}_{\geqslant 1}}}
\newcommand{\Rpos}{\ensuremath{\mathbb{R}_{\geqslant 1}}}
\newcommand{\sph}{\ensuremath{\text{type}}}
\newcommand{\qr}{\ensuremath{\text{qr}}}
\newcommand{\free}{\ensuremath{\text{free}}}
\newcommand{\N}{\ensuremath{\mathcal{N}}}
\newcommand{\A}{\ensuremath{\mathcal{A}}}
\newcommand{\B}{\ensuremath{\mathcal{B}}}
\newcommand{\C}{\ensuremath{\mathcal{C}}}
\newcommand{\T}{\ensuremath{\mathcal{T}}}
\newcommand{\TT}{\ensuremath{\mathfrak{T}}}
\newcommand{\F}{\ensuremath{\mathcal{F}}}
\newcommand{\FF}{\ensuremath{\mathfrak{F}}}
\newcommand{\CC}{\ensuremath{\mathfrak{C}}}
\newcommand{\DD}{\ensuremath{\mathfrak{D}}}
\newcommand{\size}[1]{\ensuremath{\left\lVert #1 \right\rVert}}
\newcommand{\set}[1]{\ensuremath{\{ #1 \}}}
\newcommand{\lin}[1]{\ensuremath{\mathcal{O}(#1)}}
\newcommand{\poly}[1]{\ensuremath{\text{poly}(#1)}}
\newcommand{\isdef}{\ensuremath{:=}}
\renewcommand{\phi}{\varphi}
\newcommand{\dist}{\text{dist}}
\newcommand{\Tower}{\text{Tower}}

\renewcommand{\S}{\ensuremath{\mathcal{S}}}

\DeclareMathOperator{\Root}{\ensuremath{\textnormal{root}}}
\DeclareMathOperator{\eq}{\ensuremath{\textnormal{eq}}}

\DeclareMathOperator{\iso}{\ensuremath{\textnormal{iso}}}
\DeclareMathOperator{\coreach}{\ensuremath{\textnormal{co-reach}}}

\DeclareMathOperator{\ar}{\ensuremath{\textnormal{ar}}}
\DeclareMathOperator{\bit}{\ensuremath{\textnormal{bit}}}
\newcommand{\enc}{\textnormal{enc}}
\newcommand{\sle}{\leqslant}

\makeatletter
\renewenvironment{proof}[1][\proofname]{\par
  \pushQED{\qed}%
  \normalfont \topsep6\p@\@plus6\p@\relax
  \trivlist
  \item[\hskip\labelsep
        \color{darkgray}\sffamily\bfseries
    #1\@addpunct{.}]\ignorespaces
}{%
  \popQED\endtrivlist
}
\makeatother

\begin{document}

\maketitle

\begin{abstract}
  Building on the locality conditions for first-order logic by Hanf and Gaifman, Barthelmann and Schwentick showed in 1999 that every first-order formula is equivalent to a formula of the shape $\exists x_1 \dotsc \exists x_k \forall y\,\phi$ where quantification in $\phi$ is relativised to elements of distance $\leq r$ from $y$. Such a formula will be called Barthelmann-Schwentick normal form (\BSNF) in the following. However, although the proof is effective, it leads to a non-elementary blow-up of the \BSNF\ in terms of the size of the original formula.

  We show that, if equivalence on the class of all structures, or even only finite forests, is required, this non-elementary blow-up is indeed unavoidable. We then examine restricted classes of structures where more efficient algorithms are possible. In this direction, we show that on any class of structures of degree $\leq 2$, \BSNF\ can be computed in 2-fold exponential time with respect to the size of the input formula. And for any class of structures of degree $\leq d$ for some $d\geq 3$, this is possible in 3-fold exponential time. For both cases, we provide matching lower bounds.
\end{abstract}

\section{Introduction}

First-order logic (for short: $\FO$) and its extensions are employed in many fields of theoretical computer science, as for example automata theory, descriptive complexity theory, database theory, and algorithmic meta-theorems. 

However, it is well-known that the expressive power of $\FO$ is very limited: it can only express local properties. This excludes properties that require a global grasp of the structure, as for example graph connectivity. The theorems by Hanf, by Gaifman, and by Schwentick and Barthelmann \cite{Han65,FagSV95,Gaifman,SchB99} are formalisations of the locality of $\FO$ and thus facilitate inexpressibility proofs. Moreover, each of these locality theorems gives rise to a normal form for first-order logic.

In particular, Gaifman's theorem implies that on the class of \emph{all structures}, every sentence of $\FO$ is equivalent to a Gaifman normal form (\GNF), i.e., a Boolean combination of statements of the shape
\begin{quote}
  ``There are $\geq k$ nodes whose $r$-neighbourhoods are \emph{pairwise disjoint} and which satisfy the same \FO-definable property $\phi$.''
\end{quote}

Hanf's theorem implies that for every class of structures of \emph{bounded degree}, each sentence of $\FO$ logic is equivalent to a Hanf normal form (\HNF), i.e., a Boolean combination of statements of the shape
\begin{quote}
  ``There are $\geq k$ nodes whose $r$-neighbourhoods each have isomorphism type $\tau$.''
\end{quote}
Hanf's and Gaifman's theorem have found a plethora of applications in algorithms and complexity (cf., e.g., \cite{Seese,FrickGrohe01,Libkin-FMT,DurandGrandjean,KazanaSegoufin-boundedDegree,Kreutzer-AMT-Survey,DBLP:conf/stacs/Segoufin14,DBLP:conf/stoc/GroheKS14,berkholz_et_al:LIPIcs:2017:7053,KuskeSchweikardt2017}). In particular, algorithmic meta-theorems make use of the local conditions expressed in \GNF\ and \HNF\ to show that on many classes of structures, $\FO$ model checking is fixed-parameter tractable, and that the results of $\FO$ queries can be enumerated with constant delay after a linear-time preprocessing phase. 

Schwentick and Barthelmann \cite{SchB99} presented a local normal form for first-order logic that avoids the ``pairwise disjoint $r$-neighbourhoods'' constraint in Gaifman's normal form as well as the restriction to classes of structures of \emph{bounded degree} necessary for Hanf's normal form. They showed that on the class of  \emph{all structures}, every sentence of $\FO$ is equivalent to a \emph{single} statement of the shape
\begin{quote}
  ``$\geq k$ pebbles can be placed such that the $r$-neighbourhoods of all nodes in the so extended structure satisfy the same \FO-definable property $\phi$.''
\end{quote}
In the following, we call such statements \emph{Schwentick-Barthelmann normal form} ($\BSNF$).
In~\cite{SchB99}, two applications of \BSNF\ are described:  a local variant of Ehrenfeucht-Fra\"iss\'e games~\cite{Fra54,Ehr61} which restricts the game to local neighbourhoods after an initial phase \cite{SchB99}, and an automata model for first-order logic. 

In the context of algorithmic meta-theorems, the question about the efficiency of constructing normal forms has arisen (cf., e.g., \cite{BolligKuske2012,DGKS07,Lin08,HKS13-LICS,DBLP:journals/corr/HarwathHS15,HKS16-LICS,KuskeSchweikardt2017}). In particular, it was shown in \cite{DGKS07} that there is a non-elementary lower bound for the size of \GNF\ in respect to the input sentence if equivalence is required on the class of all finite trees. On the other hand, \cite{BolligKuske2012} and \cite{HKS13-LICS} provided 3-fold exponential algorithms and matching lower bounds for the construction of \HNF\ and \GNF, respectively, on classes of structures of bounded degree.

Concerning \BSNF, the construction described in \cite{SchB99} is effective, but has non-elementary time complexity. We show (cf. Theorem~\ref{Thm:lowerBound_non-elem}) that this is indeed unavoidable -- i.e., even when equivalence of the constructed \BSNF\ to the input sentence is  required on the class of all finite forests, a non-elementary blow-up in the size cannot be avoided.

For this reason, our main focus lies on an investigation of \BSNF\ on classes of structures of \emph{bounded degree}. We show that, when equivalence is only required on the class of all structures of degree at most $d$ for $d=2$ ($d\geq 3$), any formula $\phi$ from $\FO$ logic can be turned into a \BSNF\ in $2$-fold ($3$-fold) exponential time in the size of $\phi$ (cf. Theorem~\ref{thm:bsnf-upper-bound}). We complement both upper bounds by matching lower bounds (cf. Theorem~\ref{Thm:lowerBound_d_2} and Theorem~\ref{Thm:lowerBound_d_ge2}). In particular, our upper bounds imply corresponding upper bounds on the number of pebbles to be placed in the first stage of the local Ehrenfeucht-Fra\"iss\'e game and on the size of the automata for $\FO$ logic described in \cite{SchB99} when restricting attention to classes of structures of bounded degree.

Our algorithm for the construction of \BSNF\ relies on a transformation of $\FO$ formulae into \HNF, as described in \cite{BolligKuske2012, HKS13-LICS, HKS16-LICS}. The most challenging task is to turn so-called type-formulae, which describe the isomorphism type of the $r$-neighbourhood of their free variables, into \BSNF. Our lower bound proofs use techniques already employed in \cite{StockmeyerMeyer1973,FrickGrohe-FO-MSO-revisited,FlGr06,DGKS07,BolligKuske2012,HKS13-LICS,DBLP:journals/corr/HarwathHS15}.

The rest of the paper is structured as follows. Section~\ref{section:preliminaries} fixes basic notations used throughout the paper. Section~\ref{section:upper-bounds} presents the algorithm leading to our upper bounds. Section~\ref{section:lower-bounds} provides the matching lower bounds. Due to space restrictions, some proof details are deferred to an appendix.
\subsubsection*{Acknowledgements}
The authors would like to thank Dietrich Kuske, who brought the normal form of Barthelmann and Schwentick to their attention and posed the question about its complexity. We also would like to thank Nicole Schweikardt for helpful hints on the first version of this paper.
\section{Preliminaries}
\label{section:preliminaries}
We use $\NN$ to denote the set of natural numbers, i.e., the set of nonnegative integers, and we let $\Npos \isdef \NN \setminus \set{0}$. 
For all $m,n\in\NN$ with $m\leq n$, we write $[m,n]$ for the set $\set{ i \in \NN \colon m \leq i \leq n}$ and let $[m,n]\isdef \emptyset$ if $m > n$. By $[n]$ we abbreviate the set $[1,n]$. 

$\Rpos$ is the set of all reals greater than or equal to $1$.  For a real number $r> 0$, we write $\log r$ to denote the logarithm of $r$ with respect to base $2$. 
For every function $f \colon \NN\to\Rpos$, we write $\poly{f(n)}$ for the class of all functions $g\colon\NN\to\Rpos$ for which there exists a number $c > 0$ such that $g(n) \leq (f(n))^c$ for all sufficiently large $n \in \NN$.

The function $\Tower \colon \NN \times \Rpos \to \Rpos$ is defined by $\Tower(0, x) \isdef x$ and $\Tower(k{+}1, x) \isdef 2^{\Tower(k, x)}$ for all $k\in\NN,x\in\Rpos$. I.e., $\Tower(k, x)$ is a tower of $2s$ of height $k$ with $x$ on top. We furthermore abbreviate $\Tower(n) \isdef \Tower(n, 1)$. 
A function $f\colon\NN\to\Rpos$ is \emph{at most $k$-fold exponential}, for some $k\in\NN$, if $f$ belongs to the class $\Tower(k, \poly{n})$. More generally, $f$ is \emph{elementary} if it is at most $k$-fold exponential for some $k\in\NN$ and \emph{non-elementary} if there is no such $k\in\NN$.

\subparagraph*{Signatures and Structures}
\label{section:preliminaries:signatures-and-structures}
For signatures, structures, and $\FO$ logic, we use the standard notation,  cf. \cite{EbbinghausFlum, Libkin-FMT}.
A \emph{signature} $\sigma$ is a \emph{finite} set $\set{R_1, \dotsc, R_{\ell}, c_1, \dotsc, c_k}$ of $\ell \in \NN$ relation symbols $R_1, \dotsc, R_{\ell}$ and $k\in\NN$ constant symbols $c_1,\dotsc,c_k$. Each relation symbol $R$ has an \emph{arity} $\ar(R)\in\Npos$. A $\sigma$-structure $\A$ is a tuple $(A, R_1^{\A}, \dotsc, R_{\ell}^{\A}, c_1^{\A}, \dotsc, c_{k}^{\A})$, where $A$ is a \emph{finite} and non-empty set, called the \emph{universe} of $\A$, where each relation symbol $R_i$, $i \in [\ell]$, is interpreted by the relation $R_i^{\A} \subseteq A^{\ar(R_i)}$ and where each constant symbol $c_i$, $i\in[k]$, is interpreted by an element $c_i^{\A} \in A$. We write $\A\cong\B$ to express that $\A$ is isomorphic to a second $\sigma$-structure $\B$. 

In the following, we suppose that $\sigma$ is a \emph{relational} signature, i.e., a signature that only contains relation symbols. 
A $\sigma$-structure $\A$ is a \emph{substructure} of a $\sigma$-structure $\B$ if $A \subseteq B$ and $R^{\A} \subseteq R^{\B}$ for each $R\in\sigma$. In particular, $\A$ is the \emph{substructure of $\B$ induced by $A$ (for short: $\A = \B[A]$)} if $R^{\A} = R^{\B} \cap A^{\ar(R)}$ for each $R\in\sigma$.

A $\sigma$-structure $\B$ is a \emph{disjoint union} of $s\in\Npos$ $\sigma$-structures $\A_1,\dotsc, A_s$ if the universes $A_1, \dotsc, A_s$ are pairwise disjoint, $B$ is the union of $A_1, \dotsc, A_s$, and $R^{\B}$ is the union of $R^{\A_1}, \dotsc, R^{\A_s}$ for all $R\in\sigma$. 
A $\sigma$-structure $\A$ is a \emph{component} of $\B$ if $\B$ is the disjoint union of $\A$ and some other $\sigma$-structure.

\subparagraph*{First-Order Logic}
\label{section:preliminaries:first-order-logic}
By $\FO[\sigma]$ we denote the class of all first-order formulae of signature $\sigma$. That is, $\FO[\sigma]$ is built from atomic formulae of the form $x_1{=}x_2$ and $R(x_1, \dotsc, x_{\ar(R)})$, for $R\in\sigma$ and variables or constant symbols $x_1,x_2,\dotsc,x_{\ar(R)}$, and closed under the Boolean connectives $\neg, \lor$ and existential first-order quantifiers $\exists x$ for any variable $x$.\footnote{As usual, $\forall x$, $\land$, $\impl$, $\gdw$ will be used as abbreviations when constructing formulae.} By $\FO$ we denote the union of all $\FO[\sigma]$ for arbitrary signatures $\sigma$. 

The \emph{size} $\size{\phi}$  of an $\FO[\sigma]$-formula is its length when viewed as a word over the alphabet $\sigma \cup \textup{Var} \cup \set{,} \cup \set{=,\exists,\neg,\lor, (, )}$, where \textup{Var} is a countable set of variable symbols.
The \emph{quantifier rank} $\qr(\phi)$ of an $\FO$-formula $\phi$ is defined as the maximal nesting depth of its quantifiers.
By $\free(\phi)$ we denote the set of all \emph{free variables} of $\phi$. A \emph{sentence} is a formula $\phi$ with $\free(\phi)=\emptyset$. We write $\phi(\ov{x})$, for $\ov{x} = (x_1, \dotsc, x_n)$ with $n\in\NN$, to indicate that $\free(\phi)$ is a subset of $\set{x_1,\dotsc,x_n}$.

If $\A$ is a $\sigma$-structure and $\ov{a} = (a_1,\dotsc, a_n) \in A^n$, we write $(\A,\ov{a})\models \phi(\ov{x})$ or $\A \models \phi[\ov{a}]$ to indicate that the formula $\phi(\ov{x})$ is satisfied in $\A$ when interpreting the free occurences of the variables $x_1, \dotsc, x_n$ with the elements $a_1,\dotsc,a_n$. 
Two formulae $\phi(\ov{x})$ and $\psi(\ov{x})$ over a signature $\sigma$ are \emph{$\CC$-equivalent}, for a class $\CC$ of $\sigma$-structures, if for every $\A\in\CC$ and $\ov{a}\in A^n$, we have $\A \models \phi[\ov{a}]$ if, and only if, $\A \models \psi[\ov{a}]$. In particular, we call $\phi$ and $\psi$ \emph{equivalent} if they are $\CC^{\sigma}$-equivalent for the class $\CC^{\sigma}$ of \emph{all} $\sigma$-structures.

For an $\FO[\sigma]$-formula $\phi(y)$ and a $k\in\Npos$, it is easy to define an $\FO[\sigma]$-formula $\exists^{\geq k}y\,\phi(y)$ 
of size $\lin{k^2 + \size{\phi}}$ such that $\A \models \exists^{\geq k}y\,\phi(y)$ for a $\sigma$-structure $\A$ if, and only if, there are at least $k$ elements $a$ in $A$ such that $\A \models \phi[a]$, cf., e.g., \cite{HKS13-LICS}. We write 
 $ \exists^{=k}y\,\phi(y)$ for $\exists^{\geq k}y\,\phi(y) \land \neg \exists^{\geq k+1}y\,\phi(y)$
and $\exists^{=0} y\,\phi(y)$ for $\neg \exists y\,\phi(y)$.

\subparagraph*{Gaifman Graph and Classes of Structures of Bounded Degree}
\label{section:preliminaries:gaifman-graph}
Let $\sigma$ be a signature and let $\A$ be a $\sigma$-structure.
The \emph{Gaifman graph} $G_\A$ of $\A$ is the undirected and loop-free graph with node set $A$ and
an edge between two distinct nodes $a,b\in A$ if, and only if, there is an $R \in \sigma$ and a tuple $(a_1 , \dotsc , a_{\ar(R)} ) \in R^\A$, such that $a,b \in \set{a_1, \dotsc , a_{\ar(R)}}$. 
For elements $a,b\in A$, we denote by $\dist^{\A}(a,b)$ the length of a shortest path between $a$ and $b$ in $G_{\A}$ or $\infty$ if there is no such path.
 The $\sigma$-structure $\A$ is \emph{connected} if its Gaifman graph is connected, i.e., if $\dist^{\A}(a,b) < \infty$ for all $a,b\in A$. 

 For each $d\in\NN$, an $\FO[\sigma]$-formula $\dist_{\leq d}(x, y)$
 can be constructed (cf., e.g., \cite{HKS13-LICS}) in time $\lin{\size{\sigma}\cdot\log d}$ for $d\geq 2$,
 such that for each $\sigma$-structure $\A$ and all $a,b\in A$,
\begin{equation*}
  \A \ \models \ \dist_{\leq d}[a,b] 
  \quad \text{iff} \quad
  \dist^{\A}(a,b) \leq d.
\end{equation*}

\noindent For every $r\geq 0$ and $a \in A$, the \emph{$r$-neighbourhood of $a$ in $\A$} is the set
\begin{equation*}
  N_r^{\A}(a) \ \isdef \set{b \in A \colon \dist^{\A}(a,b)\leq r}
\end{equation*}
and the $r$-neighbourhood $N_r^{\A}(\ov{a})$ of a tuple $\ov{a} = (a_1,\dotsc,a_n)\in A^n$ of length $n\in\Npos$ is the union of the sets $N_r^{\A}(a_i)$ for all $i \in [n]$. 

The \emph{degree} of $\A$ is the degree of its Gaifman graph $G_{\A}$. We say that $\A$ is \emph{$d$-bounded}, for a \emph{degree bound $d\in\NN$}, if no node in $G_{\A}$ has more than $d$ neighbours.
By $\CC^{\sigma}_d$ we denote the \emph{class of all $d$-bounded $\sigma$-structures}. Two $\FO[\sigma]$-formulae $\phi$ and $\psi$ are \emph{$d$-equivalent} if they are $\CC^{\sigma}_d$-equivalent. 
To bound the cardinality of a neighbourhood in a $d$-bounded structure in dependence from its radius $r$, let $\nu_d \colon \NN\to\NN$ be defined by
\begin{equation*}
  \nu_d(r) \ \isdef \quad 1 + d \cdot \sum_{0\leq i < r} (d-1)^i.
\end{equation*}
Then, if $\A$ is $d$-bounded, for any element $a\in A$ and any $r\in\NN$, we have $|N_r^{\A}(a)| \ \leq \ \nu_d(r)$. In particular, $\nu_0(r) = 1$, $\nu_1(r) \leq 2$, $\nu_2(r) = 2r+1$, and $(d-1)^r \leq \nu_d(r) \leq d^{r+1}$ for $d\geq 3$. I.e., $\nu_d$ is growing linearly for $d\leq 2$ and exponentially for~$d \geq 3$. 

\subparagraph*{Isomorophism Types}
\label{section:preliminaries:types}
Let $\sigma$ be a relational signature. For each $n \in \Npos$, we let $\sigma_n \isdef \sigma \cup \set{ c_1, \dotsc, c_n}$ for pairwise distinct constant symbols $c_1, \dotsc, c_n$. For any $r\in\NN$, an \emph{$r$-type (with $n$ centres over $\sigma$)} is a $\sigma_n$-structure $\tau = (\A, a_1, \dotsc, a_n)$ that consists of a $\sigma$-structure $\A$ and an interpretation $a_i \in A$ for each constant symbol $c_i$, $i \in [n]$ such that $ A = N_r^{\A}(a_1,\dotsc,a_n)$. I.e., every element of $\A$ has distance $\leq r$ to at least one of the $c_i$'s. We also call $a_1, \dotsc, a_n$ the \emph{centres of~$\tau$}. 
If $\B$ is a $\sigma$-structure, $a_1,\dotsc,a_n \in B$ for some $n\in\Npos$, and $r\in\NN$,
then $\N_r^{\B}(a_1,\dotsc,a_n)$ denotes the \emph{$r$-type $(\B[N_r^{\B}(a_1,\dotsc,a_n)], a_1, \dotsc, a_n)$ of $a_1,\dotsc,a_n$ in $\B$}.
We say that $a_1,\dotsc,a_n$ \emph{realise} an $r$-type $\tau$ with $n$ centres if $\N_r^{\B}(a_1,\dotsc,a_n) \cong \tau$. For short -- we often speak of types instead of isomorphism types.

We will often use the following observations from folklore (cf.,~\cite{HKS16-LICS, berkholz_et_al:LIPIcs:2017:7053, KuskeSchweikardt2017}):
\begin{lemma}
  \label{lem:types}
  Let $\sigma$ be a relational signature, let $d \in \NN$ with $d\geq 2$, and let $\A$ be a $d$-bounded $\sigma$-structure. For all $r\in\NN$, $n\in\Npos$, and $\ov{a} = (a_1,\dotsc,a_n)\in A^n$, it holds that:
  \begin{enumerate}
  \item $|N_r^{\A}(\ov{a})| \leq n\cdot\nu_d(r)$,
  \item if $\N_r^{\A}(\ov{a})$ is connected, then $N_r^{\A}(\ov{a}) \subseteq N_{r+(k-1)(2r+1)}^{\A}(a_i)$ for all $i\in [n]$, 
  \item given $\A$, $\ov{a}$, and $r$, the $r$-type of $\ov{a}$ in $\A$ can be computed in time $(n\cdot\nu_d(r))^{\lin{\size{\sigma}}}$, and 
  \item given two $d$-bounded $r$-types $\tau$ and $\tau'$ with $n$ centres over $\sigma$, it can be decided in time $2^{\lin{\size{\sigma}(n\cdot\nu_d(r))^2}}$ whether $\tau \cong \tau'$.
  \end{enumerate}
\end{lemma}
\begin{lemma}
  \label{lem:compute-representatives}
  There is an algorithm which upon input of a relational signature $\sigma$, a degree bound $d\in\NN$ with $d\geq 2$, a radius $r\in\NN$, and a number $n \in \Npos$, computes a set $\mathfrak{T}_r^{\sigma,d}(n)$ of $d$-bounded $r$-types with $n$ centres over $\sigma$, such that for every $d$-bounded $r$-type $\tau$ with $n$ centres over $\sigma$ there is exactly one $\tau' \in \mathfrak{T}_r^{\sigma,d}(n)$ such that $\tau \cong \tau'$. The algorithm's runtime is $2^{(n\cdot\nu_d(r))^{\lin{\size{\sigma}}}}$. 
  Furthermore, upon input of a $d$-bounded $r$-type $\tau$ with $n$ centres over $\sigma$, the particular $\tau'\in\mathfrak{T}_r^{\sigma,d}(n)$ with $\tau\cong\tau'$ can be computed in time $2^{(n\cdot\nu_d(r))^{\lin{\size{\sigma}}}}$. 
\end{lemma}

Given an $r$-type $\tau$ with $n$ centres over $\sigma$, for some $r\in\NN$ and $n\in\Npos$, one can construct a \emph{type-formula} $\sph_{\tau}(\ov{x})$ with $\ov{x} = (x_1,\dotsc,x_n)$ such that for every $\sigma$-structure $\A$ and every tuple $\ov{a} \in A^n$,
\begin{equation*}
  \A  \models  \sph_{\tau}[\ov{a}] \quad\text{if, and only if,} \quad \N_r^{\A}(\ov{a}) \cong \tau.
\end{equation*}
More precisely, if $\tau = (\B, b_1,\dotsc,b_n)$ and $B = \set{e_1,\dotsc,e_N}$ for some $N\in\Npos$, the type-formula $\sph_{\tau}(\ov{x})$ can be defined by
\label{expr:sphere-formula}
\begin{equation*}
  \begin{split}
    \exists z_1 \cdots \exists z_N\, \Biggl( 
      & \Und_{1\leq j\leq N} \Oder_{1\leq i \leq n}\!\! \dist(x_i, z_j) \leq r \ \ \und \ \ \forall z\, \Biggl( \Oder_{1\leq i\leq n}\!\!\dist(x_i, z) \leq r \ \impl \Oder_{1\leq j \leq N} z{=}z_j\Biggr) \\
      & \ \ \und \ \Und_{1\leq i\leq n}\!\! x_i{=}z_{q_i} \ \ \und \ \ \psi(z_1, \dotsc, z_N)
      \Biggr)
  \end{split}
\end{equation*}
where $\psi(z_1,\dotsc, z_N)$ is a conjunction of all atomic and negated atomic $\FO[\sigma]$-formulae $\phi(z_1,\dotsc,z_N)$ such that $\B \models \phi[e_1,\dotsc,e_N]$, and $q_1, \dotsc, q_n \in [N]$ are chosen such that $b_i = e_{q_i}$ for each $i \in [n]$. 
Thus, if $\tau$ is $d$-bounded for some $d\in\NN$, the formula 
$\sph_{\tau}(\ov{x})$ can be constructed in time $\lin{\size{\sigma}}$ if $n\cdot\nu_d(r) = 1$, and otherwise in time $(n\cdot\nu_d(r))^{\lin{\size{\sigma}}}$.

\subparagraph*{Local Formulae and the Barthelmann-Schwentick Normal Form}
\label{section:preliminaries:bsnf}
Let $\sigma$ be a relational signature. 
An $\FO[\sigma]$-formula $\phi(x_1,\dotsc,x_m, y_1,\dotsc, y_n)$  with $m\in\NN$ and $n\in\Npos$ is \emph{$r$-local around $y_1,\dotsc,y_n$} if for each $\sigma$-structure $\A$ and all elements $a_1,\dotsc, a_m, b_1, \dotsc, b_n \in A$,
\begin{equation*}
  \begin{split}
    \A \ \models \ & \phi[a_1,\dotsc,a_m,b_1,\dotsc,b_n] \\
    \quad\text{if, and only if,}\quad
    \A[ \set{a_1,\dotsc,a_m} \cup N_r^{\A}(b_1,\dotsc,b_n)] 
  \ \models \ & \phi[a_1,\dotsc,a_m,b_1,\dotsc,b_n].    
  \end{split}
\end{equation*}
We call $\phi$ \emph{local around $y_1,\dotsc,y_n$} if it is $r$-local around $y_1,\dotsc,y_n$ for some $r\in\NN$. As an example, the type-formula $\sph_{\tau}(\ov{x})$ is $r$-local around $\ov{x}$ if $\tau$ is an $r$-type.

A formula is in \emph{Barthelmann-Schwentick normal form} (for short: \BSNF) if it has the shape 
\begin{equation*}
  \exists y_1 \cdots \exists y_n \forall z\, \phi(\ov{x}, y_1, \dotsc, y_n, z)
\end{equation*}
for an $n\geq 0$ and a formula $\phi$ where every quantification is
restricted to elements of the universe of distance at most $r$ from $z$, i.e., the formula $\phi$ is $r$-local around $z$, for some $r\geq 0$ (cf., \cite{SchB99}). Its locality radius is~$r$.
A \BSNF-formula is a formula in \BSNF\ and a \BSNF-sentence is a sentence in \BSNF.

\subparagraph*{Forests and Trees}
\label{section:preliminaries:forests-and-trees}
A \emph{forest} is a directed graph where every vertex has indegree at most $1$ and
whose Gaifman graph is acyclic. A \emph{tree} is a connected forest. In forests, as well as in trees, nodes of indegree $0$ are called \emph{roots}. By $\FF$ we denote the class of all finite forests. 
The \emph{height} of a forest $\F$ (a tree $\T$) is the length of the longest path in $\F$ (in $\T$), starting in a (the) root node. $\FF_{\sle h}$ denotes the class of all finite forest with height $\le h$.

\section{Upper Bounds}
\label{section:upper-bounds}

This section's aim is to show that, in contrast to the non-elementary lower bound on trees of unbounded degree (cf. Theorem \ref{Thm:lowerBound_non-elem}), Barthelmann-Schwentick normal forms can be computed in elementary time when equivalence to the input formula is only required on a class of structures of bounded degree. The main result of this section can be stated as follows:

\begin{theorem}
  \label{thm:bsnf-upper-bound}
  There is an algorithm which, on input of a degree bound $d\in\NN$ with $d\geq 2$, a relational signature $\sigma$, and a formula $\phi(\ov{x})$ from $\FO[\sigma]$, computes a formula $\phi^B(\ov{x})$ in \BSNF\ that is $d$-equivalent to $\phi(\ov{x})$.   
  The algorithm  runs in time
  \begin{equation*}
    2^{(\size{\phi}\cdot\nu_d((n+1)(2\cdot 4^q+1)))^{\lin{\size{\sigma}}}},
  \end{equation*}
  where $n,q\geq 0$ are the number of free variables and the quantifier rank of $\phi$, respectively.
\end{theorem}
\begin{remark}
  Under the assumption that $\sigma$ only contains relation symbols that actually occur in $\phi$ and
  since $n,q < \size{\phi}$, the algorithm of Theorem~\ref{thm:bsnf-upper-bound} runs in time
  \begin{equation*}
    2^{2^{\poly{\size{\phi}}}} \quad \text{for $d=2$, and}\quad
    2^{d^{2^{\lin{\size{\phi}}}}} \quad \text{for $d\geq 3$.}
  \end{equation*}
\end{remark}
The algorithm described in Theorem~\ref{thm:bsnf-upper-bound} relies on the construction of Hanf normal forms described in \cite{BolligKuske2012, HKS13-LICS, HKS16-LICS} and proceeds in the following four steps, which are carried out in detail in the subsequent Section~\ref{section:positive-hnf}, \ref{section:counting-sentences-to-BSNF}, \ref{section:type-formulae-to-BSNF}, and \ref{section:bsnf-upper-bound}, respectively:
\begin{enumerate}
\item The input formula $\phi(\ov{x})$ is transformed into a $d$-equivalent positive Hanf normal form~$\phi^H_+(\ov{x})$. Intuitively, a positive Hanf normal form is built from the following sub-formulae using \emph{only the logical connectives $\und$ and $\oder$}:
  \begin{itemize}
  \item \emph{Counting-sentences}, which either state that there are at least $k\in\Npos$ elements that realise a given type or that there are precisely $k\in\NN$ elements that realise a given type.
  \item \emph{Type-formulae}, which check whether the interpretation of their free variables realises a given type.
  \end{itemize}
\item Each counting-sentence in $\phi^H_+(\ov{x})$ is replaced by an equivalent sentence in \BSNF.
\item Each type-formula in $\phi^H_+(\ov{x})$ is replaced by a $d$-equivalent formula in \BSNF.
\item The formula obtained from the latter two steps is a positive Boolean combination of sentences and formulae in \BSNF. We use a procedure from \cite{SchB99} to turn this positive Boolean combination into a single equivalent formula in \BSNF.
\end{enumerate}

\noindent In the remainder of this section, $\sigma$ will always denote a relational signature.

\subsection{(Positive) Hanf Normal Form}
\label{section:positive-hnf}
In this section, we recall the notion of \emph{Hanf normal form} (\HNF) from \cite{HKS16-LICS} and introduce its syntactical restriction to \emph{positve Hanf normal form ($\HNF_+$)}. 

A \emph{threshold-counting-sentence} has the shape
\begin{equation*}
  \exists^{\geq k}y\,\sph_{\tau}(y)
\end{equation*}
where $k \in\Npos$ and, for some $r\in\NN$, $\tau$ is an $r$-type with one centre. A $\sigma$-structure $\A$ satisfies $\exists^{\geq k}y\,\sph_{\tau}(y)$ if, and only if, there are $\geq k$ pairwise distinct elements $a$ in $A$ that realise $\tau$. 

An $\FO[\sigma]$-formula is in \emph{Hanf normal form} (for short: \HNF) if it is a Boolean combination of type-formulae and threshold-counting-sentences. Its \emph{locality radius} is the maximum radius of all its type-formulae.

\begin{theorem}[\cite{BolligKuske2012, HKS13-LICS, HKS16-LICS}]
  \label{thm:fo-to-hnf}
  There is an algorithm which, on input of a degree bound $d\in\NN$ with $d\geq 2$, a relational signature $\sigma$, and a formula $\phi(\ov{x})$ with quantifier rank $q\in\NN$ from~$\FO[\sigma]$, computes a formula $\phi^H(\ov{x})$ in \HNF\ that is $d$-equivalent to $\phi(\ov{x})$ and that has locality radius~$\leq 4^q$. The algorithm runs in time
  \begin{equation*}
    2^{{(\size{\phi}\cdot\nu_d(4^q))}^{\lin{\size{\sigma}}}}.
  \end{equation*}
\end{theorem}
In the following, we also consider \emph{exact-counting-sentences} of the shape $\exists^{=k}y\,\sph_{\tau}(y)$ for arbitrary $k \in \NN$. 
We will subsume exact-counting-sentences and threshold-counting-sentences under the name \emph{counting-sentences}. The reason for introducing exact-counting-sentences is that in the notion of positive Hanf normal form, introduced in the following, negations are only allowed inside of type-formulae and counting-sentences, but not in the Boolean combination that connects these.

A formula is in \emph{positive Hanf normal form} (for short: $\HNF_+$) if it is a \emph{positive} Boolean combination (i.e., a Boolean combination that only uses the connectives $\land$ and $\lor$) of type-formulae and counting-sentences. $\HNF_+$ can be obtained from $\HNF$:

\begin{lemma}
  \label{lem:fo-to-hnfp}
  There is an algorithm which, on input of a degree bound $d\in\NN$ with $d\geq 2$, a relational signature $\sigma$, and a formula $\phi(\ov{x})$ with quantifier rank $q\in\NN$ from $\FO[\sigma]$, computes a formula $\phi^H_+(\ov{x})$ in $\HNF_+$ that is $d$-equivalent to $\phi(\ov{x})$ and that has locality radius $\leq 4^q$. 
  The algorithm runs in time
  \begin{equation*}
    2^{(\size{\phi}\cdot\nu_d(4^q))^{\lin{\size{\sigma}}}}.
  \end{equation*}
\end{lemma}

\begin{proof}
  On input of a degree bound $d\in\NN$ with $d\geq 2$, a relational signature $\sigma$, and a formula $\phi(\ov{x})$ from $\FO[\sigma]$ with quantifier rank $q\in\NN$ and $n\in\NN$ free variables, the algorithm proceeds as follows:
  \begin{enumerate}
  \item\label{lem:fo-to-hnfp:step1}
    Using the algorithm described in Theorem~\ref{thm:fo-to-hnf}, $\phi(\ov{x})$ is turned into a $d$-equivalent formula~$\phi^H(\ov{x})$ in \HNF\ that has
    locality radius $\leq 4^q$. 
  \item\label{lem:fo-to-hnfp:step2}
    Using de Morgan's law and the elemination of double negations, $\phi^H(\ov{x})$ is turned into a Boolean combination of type-formulae and threshold-counting-sentences whose negations only occur directly in front of a threshold-counting-sentence or a type-formula.
  \item\label{lem:fo-to-hnfp:step3}
    In the formula just constructed, we replace each negated type-formula by a $d$-equivalent positive Boolean combination of type-formulae and each negated threshold-counting-sentence by an equivalent positive Boolean combination of exact-counting-sentences:
    \begin{itemize}
    \item Each sub-formula of the shape $\neg\exists^{\geq k}y\, \sph_{\tau}(y)$ is equivalently replaced by the disjunction of all exact-counting-sentences $\exists^{=i}y\,\sph_{\tau}(y)$ for all $i \in [0,k{-}1]$. 
    \item\label{lem:fo-to-hnfp:step4}
      Consider a sub-formula of the shape $\neg\sph_{\rho}(\ov{x}')$,
      where, for an $r\in\NN$, $\rho$ is an $r$-type with $n'\in [1,n]$ centres and  $\ov{x}'$ is a sub-tuple of $\ov{x}$ of length $n'$. This formula is $d$\hbox{-}equivalent to the disjunction of all type-formula $\sph_{\rho'}(\ov{x}')$ for all $\rho' \in \mathfrak{T}^{\sigma,d}_r(n')$ with $\rho'\not\cong\rho$ (recall from Lemma~\ref{lem:compute-representatives} that $\mathfrak{T}^{\sigma,d}_r(n')$ is a set of representatives of the isomorphism classes of all $d$-bounded $r'$-types with $n'$ centres over~$\sigma$). 
    \end{itemize}
  \end{enumerate}
  Clearly, the resulting formula also has locality radius $\leq 4^q$. The time complexity of the algorithm is determined by the upper bounds provided by Theorem~\ref{thm:fo-to-hnf} and Lemma~\ref{lem:compute-representatives}. A detailed analysis is deferred to  Appendix~\ref{appendix:upper-bounds:positive-hanf-normal-form}. 
\end{proof}
In the following two sections, we will describe how counting-sentences and type-formulae can be turned into \BSNF-formulae. 

\subsection{From Counting-Sentences to BSNF}
\label{section:counting-sentences-to-BSNF}

In this section, we show that every counting-sentence $\chi$ can be turned into a \BSNF-sentence that is equivalent to $\chi$ -- not only on a class of structures of bounded degree, but on the class of all structures.

\begin{lemma}\label{lem:counting-sentences-to-BSNF}
  There is an algorithm which, on input of a counting-sentence $\chi$ computes an equivalent sentence in \BSNF\ in time $\lin{\size{\chi}}$. 
\end{lemma}

\begin{proof}
  The algorithm distinguishs on the possible shapes of the counting-sentence $\chi$. In each case, it is easy to verify that the provided \BSNF-sentence is indeed equivalent to $\chi$.
  \begin{itemize}
  \item A threshold-counting-sentence $\exists^{\geq k}y\,\sph_{\tau}(y)$ with $k\in\Npos$ is equivalent to the \BSNF-sentence
    \begin{equation*}
      \exists y_1 \cdots \exists y_k \forall z
      \left(
        \Und_{1\leq i < j \leq k} \neg y_i{=}y_i 
        \ \und 
        \left(
          \left(
        \Oder_{1\leq i\leq k}
          z{=}y_i
          \right) \ \impl \ \sph_{\tau}(z)
        \right)
      \right).
    \end{equation*}
  \item An exact-counting-sentence $\exists^{=k}y\,\sph_{\tau}(y)$ with $k\in\Npos$ is equivalent to the \BSNF-sentence
    \begin{equation*}
      \exists y_1 \cdots \exists y_k \forall z
      \left(
        \Und_{1\leq i < j \leq k} \neg y_i{=}y_j
        \ \und \ 
        \left(
          \sph_{\tau}(z) \ \gdw  
          \Oder_{1\leq i\leq k} z{=}y_i
        \right)
      \right).
    \end{equation*}
  \item An exact-counting-sentence $\exists^{=0}y\,\sph_{\tau}(y)$ is equivalent to the \BSNF-sentence
    $\forall z\, \neg \sph_{\tau}(z)$.
  \end{itemize}
  The analysis of the time complexity of the algorithm boils down to an analysis of the size of the constructed \BSNF-sentence and is deferred to  Appendix~\ref{appendix:upper-bounds:from-counting-sentences-to-BSNF}. 
\end{proof}

\subsection{From Type-formulae to BSNF}
\label{section:type-formulae-to-BSNF}
Aim of this section is to turn type-formulae into $d$-equivalent formulae in \BSNF.

\begin{lemma}\label{lem:type-formulae-to-BSNF}
  There is an algorithm which, on input of a degree bound $d\in\NN$ with $d\geq 2$, a relational signature $\sigma$, and a type-formula $\alpha(\ov{x}) \isdef \sph_{\rho}(\ov{x})$, where, for some $r\in\NN$ and $n\in\Npos$, $\rho$ is an $r$-type with $n$ centres over $\sigma$, computes a formula $\alpha^B(\ov{x})$ in \BSNF\ that is $d$-equivalent to $\alpha(\ov{x})$ and that has locality radius $n(2r{+}1)$.
  The algorithm runs in time
  \begin{equation*}
    2^{\nu_d(n(2r{+}1))^{\lin{\size{\sigma}}}}.
  \end{equation*}
\end{lemma}

\begin{proof}
  We describe the algorithm on input of a degree bound $d\in\NN$ with $d\geq 2$, a relational signature~$\sigma$, and a type-formula $\alpha(\ov{x}) \isdef \sph_{\rho}(\ov{x})$. Let $r\in\NN$ and $n\in\Npos$ be the radius and the number of centres of $\rho$, respectively. I.e., $\ov{x} = (x_1,\dotsc, x_n)$ and $\rho$ has the shape $(\A, a_1, \dotsc, a_n)$ for a $\sigma$-structure $\A$ and centres $a_1, \dotsc,a_n \in A$ such that $A = N_r^{\A}(a_1,\dotsc,a_n)$.

  Suppose that $\A_1, \dotsc, \A_k$, for a suitable $k\in\Npos$, are the connected components of $\A$. Then, each of the centres $a_1,\dotsc,a_n$ of $\rho$ belongs to precisely one of the structures $\A_1, \dotsc, \A_k$, and each of the structures $\A_1, \dotsc, \A_k$ contains at least one of these centres. 
  For each $i \in [k]$, 
  \begin{itemize}
  \item let $\ell_i \in [1,n]$ denote the number of centres among $a_1,\dotsc,a_n$ that belong to $\A_i$, and
  \item let $P_i \isdef \set{ p_{i,1}, \dotsc, p_{i,\ell_i}} \subseteq [n]$ with $|P_i|=\ell_i$ be the non-empty set of indices of all centres that belong to $\A_i$, i.e., such that $j \in P_i$ if, and only if, $a_j \in A_i$, and
  \item let $\rho_i \isdef (\A_i, a_{p_{i,1}}, \dotsc, a_{p_{i,\ell_i}})$.
  \end{itemize}
  Note that, as a consequence of Lemma~\ref{lem:types}, $A_i \subseteq N_{r_i}^{\A}(a_{p_{i,1}})$ for each $i\in[k]$ when choosing 
  \begin{equation*}
    r_i \ \isdef \quad \ell_i (2r{+}1) \ > \ r+(\ell_i{-}1)(2r{+}1).    
  \end{equation*}
  Furthermore, if $i,j\in[k]$ are distinct and $p \in P_i$, $q\in P_j$, then $a_p$ and $a_q$ have distance $\geq 2r{+}1$ in $\A$. I.e., $N_r^{\A}(a_p)$ and $N_r^{\A}(a_q)$ do not intersect and there are no edges in the Gaifman graph of $\A$ between elements of the two sets (see Figure~\ref{fig:type-formula-to-BSNF-1} for an illustration). 
 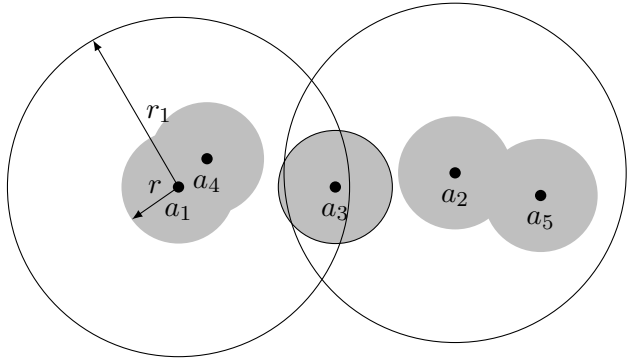
\begin{figure}
    \centering
    \begin{tikzpicture}[scale=0.75]
      \draw[clip,draw=none] (-4.1,-3.01) rectangle (7, 3.31);
      \coordinate (a1) at (-1, 0);
      \coordinate (a4) at (-0.5, 0.5);
      \coordinate (a3) at (1.75, 0);
      \coordinate (a2) at (3.85, 0.25);
      \coordinate (a5) at (5.35, -0.15);
      \draw[draw=none,fill=lightgray] (a1) circle (1);
      \draw[draw=none,fill=lightgray] (a2) circle (1);
      \draw[draw=none,fill=lightgray] (a3) circle (1);
      \draw[draw=none,fill=lightgray] (a4) circle (1);
      \draw[draw=none,fill=lightgray] (a5) circle (1);
      \draw[draw=black] (a1) circle (3);
      \draw[draw=black] (a2) circle (3);
      \draw[draw=black] (a3) circle (1);
      \node[fill=black, inner sep=1.5,circle, label=below:$a_1$] at (a1) {};
      \node[fill=black, inner sep=1.5,circle, label=below:$a_2$] at (a2) {};
      \node[fill=black, inner sep=1.5,circle, label=below:$a_3$] at (a3) {};
      \node[fill=black, inner sep=1.5,circle, label=below:$a_4$] at (a4) {};
      \node[fill=black, inner sep=1.5,circle, label=below:$a_5$] at (a5) {};
      \draw[-latex,black] (a1) -- +(-145:1) node[midway,above] { $ r $ };
      \draw[-latex,black] (a1) -- +(120:3) node[midway,right] { $ r_1$ };
    \end{tikzpicture}
    \caption{Example for the distribution of the centres of $\rho$ and connected components of $\A$ for $n = 5$, $k = 3$ and $P_1 = \set{1,4}$, $P_2 = \set{3}$, $P_3 = \set{2,5}$.}
    \label{fig:type-formula-to-BSNF-1}
  \end{figure}

  The following claim is the crucial step in constructing the desired \BSNF-formula. It shows that for each $i \in [k]$ there is a ``type-formula'' for the $r$-type $\rho_i$ which is $r_i$-local around the single variable $x_{p_{i,1}}$ (instead of being local around all its free variables) and which additionally verifies that the centres belonging to the other connected components of $\A$ are sufficiently far away.
  \begin{claim}
    \label{claim:type-formulae-to-BSNF}
    Let $ i \in [k]$. There is an $\FO[\sigma]$-formula $\gamma_i(\ov{x})$ which is
    $r_i$-local around the single free variable~$x_{p_{i,1}}$ such that
    the following holds:
    If $\C$ is a $d$-bounded $\sigma$-structure and $\ov{c} = (c_1, \dotsc, c_n) \in C^n$, then $\C \models \gamma_i[\ov{c}]$ if, and only if, the following conditions are satisfied:
    \begin{enumerate}
    \item\label{claim:type-formulae-to-BSNF:1} $\N_r^{\C}(c_{p_{i,1}}, \dotsc, c_{p_{i,\ell_i}}) \cong \rho_i$  and
    \item\label{claim:type-formulae-to-BSNF:2} for each $q \in [n]\setminus P_i$, the sets $N_r^{\C}(c_q)$ and $N_r^{\C}(c_{p_{i,1}}, \dotsc, c_{p_{i,\ell_i}})$ do not intersect and there are no edges in the Gaifman graph of $\C$ between elements of the two sets.
    \end{enumerate}
  \end{claim}
  Before proving Claim~\ref{claim:type-formulae-to-BSNF}, let us first show how these local formulae can be used to obtain the final \BSNF-formula $\alpha^B(\ov{x})$ that is $d$-equivalent to $\sph_{\rho}(\ov{x})$. 

  For this, let $y$ be a variable that is distinct to each of the variables in the tuple $\ov{x}$.
  For each $i \in [1,k]$, let $\ov{x}_i$  be the tuple of variables obtained from $\ov{x}$ by replacing the variable $x_{p_{i,1}}$ with $y$. Then, $\alpha_B(\ov{x})$ can be chosen as the formula
  \begin{equation*}
    \alpha^B(\ov{x}) \ \isdef \quad
    \forall y \Und_{i=1}^k 
    \Bigl(
      y{=}x_{p_{i,1}} \ \impl \ 
      \gamma_{i}(\ov{x}_i)
    \Bigr).
  \end{equation*}
  In particular, note that each of the formulae $\gamma_i(\ov{x}_i)$ for $i\in[k]$ is $r_i$-local around $y$ and thus, the whole universally quantified subformula of $\alpha_B(\ov{x})$ is $n(2r{+}1)$-local around $y$. 
  The $d$-equivalence of $\alpha^B(\ov{x}$) to $\sph_{\rho}(\ov{x})$ can easily be shown using the assumption on the shape of $\rho$ and the two conditions of Claim~\ref{claim:type-formulae-to-BSNF}.
  \smallskip

  \noindent\emph{We now turn to the proof of Claim~\ref{claim:type-formulae-to-BSNF}.} 
  Let $i \in [k]$ and let $\mathfrak{T}_i$ denote a set  of representatives of the isomorphism classes of all $d$-bounded $r_i$-types $\tau = (\B, b_1, \dotsc, b_{\ell_i})$ with  $\ell_i$ centres such that 
  $B \subseteq N_{r_i}^{\B}(b_1)$, i.e., all elements have distance at most $r_i$ from $b_1$, and  $\N_r^{\B}(b_1, \dotsc, b_{\ell_i}) \cong \rho_i$.
  \begin{claim}
    \label{claim:type-formulae-to-BSNF-2}
    For each $r_i$-type $\tau = (\B,b_1,\dotsc,b_{\ell_i}) \in \mathfrak{T}_i$, there is an $\FO[\sigma]$-formula $\gamma_{i,\tau}(x_1, \dotsc, x_n)$ that is 
    $r_i$-local around the free variable $x_{p_{i,1}}$ and 
    for each $d$-bounded $\sigma$-structure $\C$ and all $c_1,\dotsc,c_n \in C$, it holds that $\C \models \gamma_{i,\tau}[c_1,\dotsc,c_n]$ if, and only if, the following conditions hold:
    \begin{enumerate}
    \item $\N_r^{\C}(c_{p_{i,1}}, \dotsc, c_{p_{i,\ell_i}}) \cong \N_r^{\B}(b_1,\dotsc,b_{\ell_i}) \ ( \cong \rho_i)$
    \item Condition (2) of Claim~\ref{claim:type-formulae-to-BSNF} holds. I.e., for each element $c_q$ with $q \in [n]\setminus P_i$, the distance of $c_q$ to each of the elements $c_{p_{i,1}}, \dotsc, c_{p_{i,\ell_i}}$ is at least $2r{+}1$. 
    \end{enumerate}
  \end{claim}
  Observe that, by definition of the set $\mathfrak{T}_i$, for each $d$-bounded $\sigma$-structure $\C$ and every tuple $\ov{c} \in C^n$, $\N_r^{\C}(c_{p_{i,1}}, \dotsc, c_{p_{i, \ell_i}})\cong\rho_i$ if, and only if, there is an $r_i$-type $\tau \in \mathfrak{T}_i$ such that $\N_r^{\C}(c_{p_{i,1}},\dotsc,c_{p_{i,\ell_i}})\cong\tau$. Thus, using the formulae provied by Claim~\ref{claim:type-formulae-to-BSNF-2}, 
  we can let $\gamma_i(\ov{x})$ be the disjunction of the formulae $\gamma_{i,\tau}(\ov{x})$ for all $\tau \in \mathfrak{T}_i$. 
  In particular, $\gamma_i(\ov{x})$ is $r_i$-local around~$x_{p_{i,1}}$ since already each of the formulae $\gamma_{i,\tau}(\ov{x})$ for $i\in[k]$ is $r_i$-local around $x_{p_{i,1}}$.

  \smallskip

  \noindent\emph{For the proof of Claim~\ref{claim:type-formulae-to-BSNF-2},} 
  consider an $r_i$-type $\tau = (\B, b_1, \dotsc, b_{\ell_i})$ from  $\mathfrak{T}_i$.
  The formula $\gamma_{i,\tau}(\ov{x})$ extends the formula $\sph_{\tau}(x_{p_{i,1}},\dotsc,x_{p_{i,\ell_i}})$ from Page~\pageref{expr:sphere-formula}.
  Suppose that $ B = \set{e_1,\dotsc,e_N}$ for a suitable $N \in \Npos$.  Let $\psi(z_1, \dotsc, z_N)$ be a conjunction of all atomic and negated atomic formulae $\phi(z_1, \dotsc, z_N)$ over the signature $\sigma$ such that $\B \models \phi[e_1,\dotsc,e_N]$. Furthermore, choose $q_1, \dotsc, q_{\ell_i} \in [N]$ such that $b_j = e_{q_j}$ for all $j \in [\ell_i]$. With this, we define $\gamma_{i,\tau}(\ov{x})$ as
  \begin{equation*}
    \begin{split}
      \exists z_1 \cdots \exists z_N
      \Biggl( \ \ \
      & \Und_{\mathclap{1\leq j\leq N}} \dist(x_{p_{i,1}}, z_j) \leq r_i \ \ \und \ \ \forall z \Bigl( \dist(x_{p_{i,1}}, z) \leq r_i \ \ \impl \ \ \Oder_{\mathclap{1\leq j\leq N}} z{=}z_j \Bigr) \\
      & \und \ \ \ \Und_{\mathclap{1\leq j \leq \ell_i}} x_j{=}z_{q_j} \ \ \und \ \ \psi(z_1, \dotsc, z_N) \\
      & \und \ \ \ \Und_{\mathclap{1\leq j \leq \ell_i}} \forall z 
      \Bigl(
      \dist(z_{q_j}, z) \leq 2r{+}1 \ \ \impl \ \ \neg
      \Oder_{\mathclap{m \in [n]\setminus P_i}} z{=}x_m
      \Bigr)
      \Biggr).
    \end{split}
  \end{equation*}
  The first line replaces the first line of the formula $\sph_{\tau}$ (recall that all of $\tau$ is contained in the $r_i$-neighbourhood of its first center $b_1$) and ensures that $\gamma_{i,\tau}$ is $r_i$-local around $x_{p_{i,1}}$. Using line one, the second line ensures that an interpretation of the free variables $x_{p_{i,1}}, \dotsc, x_{p_{i,\ell_i}}$ realises the $r_i$-type $\tau$ and thus, by choice of $\tau$, $\gamma_{i,\tau}(\ov{x})$ satisfies Condition (1) of Claim~\ref{claim:type-formulae-to-BSNF-2}. The third line finally ensures that $\gamma_{i,\tau}(\ov{x})$ satisfies Condition (2) of Claim~\ref{claim:type-formulae-to-BSNF-2}. 
  As $2r{+}1 \leq r_i$, the formula $\gamma_{i,\tau}(\ov{x})$ is in particular $r_i$-local around $x_{p_{i,1}}$. This completes the proof of Claim~\ref{claim:type-formulae-to-BSNF-2} and also the proof of Claim~\ref{claim:type-formulae-to-BSNF}.
  \smallskip

  The time complexity of the algorithm is determined by the construction of the sets $\mathfrak{T}_i$ according to Lemma~\ref{lem:compute-representatives}. A detailed analysis can be found in Appendix~\ref{appendix:upper-bounds:from-sphere-formulae-to-BSNF}.
\end{proof}

\subsection{Proof of Theorem~\ref{thm:bsnf-upper-bound}}
\label{section:bsnf-upper-bound}
This section is devoted to the proof of the Theorem~\ref{thm:bsnf-upper-bound}, which combines the results of Section~\ref{section:positive-hnf}, Section~\ref{section:counting-sentences-to-BSNF}, and Section~\ref{section:type-formulae-to-BSNF} and uses a technique from \cite{SchB99} to turn a positive Boolean combination of \BSNF-formulae into a single \BSNF-formula.

\begin{proof}[Proof of Theorem~\ref{thm:bsnf-upper-bound}]
  On input of a degree bound $d\geq 2$, a relational signature $\sigma$, and a formula $\phi(\ov{x})$ from $\FO[\sigma]$ with quantifier rank $q\geq 0$ and the $n\geq 0$ free variables $\ov{x}$, the algorithm proceeds as follows:

  \begin{enumerate}
  \item Use Lemma~\ref{lem:fo-to-hnfp}
    to turn $\phi(\ov{x})$ into a $d$-equivalent formula $\phi^H_+(\ov{x})$ in $\HNF_+$.
  \item
    Use Lemma~\ref{lem:counting-sentences-to-BSNF} to replace the counting-sentences in $\phi^H_+(\ov{x})$ by equivalent \BSNF-formulae.
  \item 
    Use Lemma~\ref{lem:type-formulae-to-BSNF} to replace the type-formulae in $\phi^H_+(\ov{x})$ by $d$-equivalent \BSNF-formulae.
  \item Let $\psi(\ov{x})$ the positive Boolean combination of \BSNF-formulae, obtained in the last steps. An argument from \cite{SchB99} uses an induction over the Boolean combination to transform~$\psi(\ov{x})$ into a single \BSNF-formula:
    Suppose that $\psi(\ov{x})$ has the shape
    \begin{equation*}
      \underbrace{
        \exists y_1 \cdots \exists y_{m} \forall z \phi(\ov{x}, y_1, \dotsc, y_{m}, z)}_{\text{\BSNF-formula}}
      \quad \star \quad
      \underbrace{
        \exists y_1' \cdots \exists y'_{m'} \forall z' \phi'(\ov{x}, y'_1, \dotsc, y'_{m'}, z')}_{\text{\BSNF-formula}},
    \end{equation*}
    where $\star \in \set{\und,\oder}$ is a Boolean connective. We proceed along the following case distinction:
    \begin{itemize}
    \item If $\star = \und$, then $\psi(\ov{x})$ is equivalent to the \BSNF-formula
      \begin{equation*}
        \exists y_1 \cdots \exists y_m 
        \exists y_1' \cdots \exists y'_{m'}
        \forall z 
        \Bigl( \phi(\ov{x}, y_1, \dotsc, y_m, z) \ \und \ \phi'(\ov{x}, y'_1, \dotsc, y'_{m'}, z) \Bigr).
      \end{equation*}
    \item If $\star = \oder$, then $\psi(\ov{x})$ is equivalent to the \BSNF-formula
      \begin{equation*}
        \begin{split}
          \exists y \exists y'
          \exists y_1 \cdots \exists y_m
          \exists y'_1 \cdots \exists y'_{m'}
          \forall z
          \Bigl( &
          \bigl(
          y{=}y' \ \und \ \phi(\ov{x}, y_1, \dotsc, y_m, z)
          \bigr) \\
          & \ \oder \
          \bigl(
          \neg y{=}y' \ \und \ \phi'(\ov{x}, y'_1, \dotsc, y'_{m'}, z)
          \bigr)
          \Bigr).          
        \end{split}
      \end{equation*}
    \end{itemize}
  \end{enumerate}
  The time complexity of the described algorithm is largely determined by the time complexity of the algorithms of the employed lemmata. A detailed analysis is deferred to Appendix~\ref{appendix:upper-bounds:proof-of-main-theorem}.
\end{proof}
After proving our upper bounds on the construction of \BSNF\ on classes of structures of bounded degree, the subsequent Section~\ref{section:lower-bounds} is devoted to lower bounds. In particular, we will show that our upper bounds are worst-case optimal.
\section{Lower Bounds}
\label{section:lower-bounds}

In this section, we provide lower bounds 
on the size
of \BSNF-formulae for various classes $\CC$ of structures. 
To this aim, we present slow-growing sequences of formulae $\phi$ over suitable signatures with corresponding lower bounds on the size of $\CC$-equivalent \BSNF-formulae.
Note that similar methods as used here go back to \cite{StockmeyerMeyer1973} and were also
applied in \cite{FrickGrohe-FO-MSO-revisited, DBLP:conf/lics/PanV06} for lower bounds in parameterised complexity theory,
in \cite{grohe_schweikardt_2004, DBLP:journals/lmcs/GroheS05} for lower bounds on the succinctness of logics, and in \cite{DGKS07,BolligKuske2012,HKS13-LICS,DBLP:journals/corr/HarwathHS15} to obtain lower bounds on the size of Gaifman normal form, Hanf normal form,  Feferman-Vaught decompositions, and existential(-positive) sentences.

\subsection{Lower Bounds for Structures of Unbounded Degree}

This section's aim is the proof of the following lower bound on the size of \BSNF-sentences on classes $\FF_{\sle h}$ of forests of unbounded degree with height $\leq h$.

\begin{theorem} \label{Thm:lowerBound_non-elem}
  There is a sequence $(\phi_h)_{h\geq 1}$ of $\FO[\set{E}]$-sentences of size $\lin{h}$ such that for each $h\geq 1$, every sentence in \BSNF\ that is equivalent to $\phi_h$ on $\FF_{\sle h}$ has size $> \Tower(h)$.
\end{theorem}

\noindent To achieve this result, we will recall and use tree encodings of natural numbers 
described in~\cite{FlGr06,DGKS07}. 
This will allow us
to compare ``large'' numbers in these encodings by ``small'' $\FO$-formulae, 
that do not have ``small'' equivalent \BSNF-formulae.

For natural numbers $i,n$, we write $\bit(i,n)$ to denote the $i$-th bit in the binary representation of $n$. I.e., $\bit(i,n) \isdef 0$, if $\lfloor \frac{n}{2^i} \rfloor$ is even, and  $\bit(i,n) \isdef 1$ otherwise.
The \emph{tree encoding $\T(i)$ of an $i\in\NN$} is a tree over the signature $\set{E}$, defined inductively as follows: 
$\T(0)$ is the one-node tree and for $i \ge 1$, the tree $\T(i)$ is obtained by creating a new root and attaching to it all trees $\T(j)$ for all $j \in \NN$ such that $\bit(j,i)=1$.
 By induction it can easily be shown that for each $h \in \NN$, all tree encodings $\T(j)$ with $j<\Tower(h)$ have height $\leq h$. 

For $h \in \Npos$, we let $\TT^{\enc}_{\sle h}$ be the set of all tree encodings $\T(i)$ for all $i \in \NN$ with $i < \Tower(h)$, i.e., where $\T(i)$ has height $\leq h$.
Note that $|\TT^{\enc}_{\sle h}| = \Tower(h)$. 
By $\FF^{\enc}_{\leq h}$ we denote the class of all finite forests where every component is isomorphic to a tree encoding from $\TT^{\enc}_{\sle h}$. 
In particular, that means that a forest $\F\in\FF^{\enc}_{\leq h}$ may contain more than one component that is isomorphic to the same tree encoding from $\TT^{\enc}_{\sle h}$.

The following lemma from \cite{FlGr06, DGKS07} shows that tree encodings of numbers can be compared by “small” $\FO[\set{E}]$-formulae.

\begin{lemma}[{\cite[Lemma 10.21]{FlGr06}}]\label{lem:eq-formula} There is a number $c \in \Npos$ and a sequence $(\eq_h(x, y))_{h \geq 1}$ of $\FO[\set{E}]$-formulae of size $\leq c \cdot h$,  such
  that for each $h \in \Npos$, every forest $\F \in  \FF$, and all nodes $a, b \in F$ , the following
  holds: If there are $i, j \in [0,{\normalfont \Tower(h)}{-}1]$ such that $\S^\F(a) \cong \T(i)$ and $\S^\F(b) \cong \T(j)$ then
  \begin{equation*}
    \F \models \eq_h[a, b] \quad \text{ if, and only if, } \quad i = j.
  \end{equation*}
\end{lemma}
\noindent The core of the proof of Theorem \ref{Thm:lowerBound_non-elem} is the following combinatorial lemma which, for suitable classes $\CC$ of structures and $\FO$-sentences, provides a lower bound on the size of $\CC$-equivalent \BSNF-sentences. We will also use the lemma later on in our proofs for lower bounds on classes of structures of bounded degree.

\begin{lemma}\label{lem:combinatorial-argument}
Let $\sigma$ be a relational signature and let $\DD$ be a finite set of pairwise non-isomorphic connected $\sigma$-structures. Let $\CC$ be the class of all $\sigma$-structures where each connected component is isomorphic to a structure from $\DD$. Suppose that there is an $\FO[\sigma]$-sentence~$\psi$ such that for any $\A \in \CC$, $\A \models \psi$ if, and only if,
every connected component of $\A$ is isomorphic to another one.

  Then, every \BSNF-sentence that is $\CC$-equivalent to $\psi$ has size~$> |\mathfrak{D}|$. 
\end{lemma}

\begin{proof}
  For a contradiction, assume that there is a sentence $\psi^B$ in \BSNF\ that is $\CC$-equivalent to~$\psi$ and that has size $\leq |\mathfrak{D}|$. 

  Let $\A \in \CC$ such that, for each $\B \in \DD$, $\A$ contains exactly two connected components that are isomorphic to $\B$. 

  Hence, $\A\models\psi$ and, by $\CC$-equivalence of $\psi$ and~$\psi^B$, also

  \begin{equation}
    \label{expr:combinatorial-argument:A-psiB}
    \A \ \models \ \psi^B.
  \end{equation}
  Since $\psi^B$ is in \BSNF, it has the shape $ \exists y_1 \cdots \exists y_n \forall z\, \phi(y_1, \dotsc, y_n, z) $

  for some $n\in \NN$ with $n < |\mathfrak{D}|$ and a formula $\phi(y_1, \dotsc, y_n, z)$ that is $r$-local around the free variable $z$ for some~$r\in \NN$. 

  By (\ref{expr:combinatorial-argument:A-psiB}), there are elements $a_1, \dotsc, a_n$ in $A$ such that
  $ \A \models \phi[a_1,\dotsc,a_n,b]$
  for every $b \in A$. 
  Since $n < |\mathfrak{D}|$, there has to be a structure $\B$ in $\mathfrak{D}$ such that none of the elements $a_1, \dotsc, a_n$ belongs two any of the two 
 connected components of $\A$ that are isomorphic to $\B$. Let $\A'$ denote the substructure of $\A$ obtained by deleting one of the components that isomorphic to~$\B$. Then, by definition of $\psi$, 
  \begin{equation}
   \label{expr:combinatorial-argument:Ap-psi}
    \A' \ \not\models \ \psi.
  \end{equation}
Note that for every element $b$ of $\A$ that is still present in $\A'$, the 
 connected component it belongs to remains unchanged. In particular, it holds that  
 $ \N_r^{\A}(b) \ \cong \ \N_r^{\A'}(b) $
  for every $b \in A'$. Thus, since $\phi$ is $r$-local around its free variable $z$, we can conclude that
  $
    \A' \ \models \ \phi[a_1, \dotsc, a_n, b]$
  for every $b \in A'$, and therefore $
    \A' \ \models \ \psi^B$. 
    
    Since $\psi$ and $\psi^B$ are assumed to be $\CC$-equivalent, it follows that
  $
    \A' \ \models \ \psi.
  $
 However, this is a contradiction to (\ref{expr:combinatorial-argument:Ap-psi}).
\end{proof}

\noindent  We are now ready to prove Theorem \ref{Thm:lowerBound_non-elem}.
\begin{proof}[Proof of Theorem \ref{Thm:lowerBound_non-elem}]
For each $h\in\Npos$, let 
\begin{equation*}
  \phi_h \ \isdef \quad \forall x \, \Bigl( \Root(x) \ \ \impl \ \ \exists y \,\bigl( \Root(y) \ \land \ \lnot x{=}y \ \land \ \eq_h(x,y) \bigr) \Bigr)
\end{equation*}
where $\Root(x) \isdef \lnot \exists y \, E(y,x)$ and where the formula $\eq_h(x,y)$ is provided by Lemma~\ref{lem:eq-formula}. By Lemma~\ref{lem:eq-formula}  there is a number $c \in \Npos$ such that the size of $\phi_h$ is $\leq c \cdot h$ for each $h \ge 1.$

For the following, we fix an $h \in \Npos$.
Recall that each $\F \in \FF^{\enc}_{\sle h}$ is a disjoint union of trees that are each isomorphic to some tree encoding from $\TT^{\enc}_{\sle h}$. In particular,
\begin{equation*}
     \F \models \phi_h \quad
    \text{if, and only if,} \quad  \text{every tree in $\F$ is isomorphic to another one.}
\end{equation*}

Since $|\TT^{\enc}_{\sle h}| = \Tower(h)$, it follows from Lemma~\ref{lem:combinatorial-argument} that every \BSNF-sentence that is equivalent to $\phi_h$ on $\FF_{\sle h}$ has size $> \Tower(h)$. This completes the proof of Theorem \ref{Thm:lowerBound_non-elem}.
\end{proof}

\subsection{Lower Bounds for Structures of Bounded Degree}

This section shows that our algorithm (cf. Theorem~\ref{thm:bsnf-upper-bound}) is worst-case optimal. To this aim, we provide a 2-fold exponential lower bound on the size of $\BSNF$-formulae on classes of structures of degree $\le 2$, and 3-fold exponential lower bounds on classes of degree $\leq d$ for~$d \geq3 $.

We first prove the lower bound for degree $\leq 2$:

\begin{theorem}\label{Thm:lowerBound_d_2}
 There is a signature $\sigma$ and a sequence $(\phi_h)_{h\geq 1}$ of $\FO[\sigma]$-sentences of size~$\lin{h}$  such that for each $h\geq 1$, every sentence in \BSNF\ that is equivalent to $\phi_h$ on all structures of degree $\leq 2$ over the signature $\sigma$ has size $> 2^{2^h+1}$.
\end{theorem}

To this aim, we use the same approach as in the proof of Theorem \ref{Thm:lowerBound_non-elem}. Restricting the class of all trees to structures of degree $2$, we obtain the class of paths, which are also called \emph{chains}.
In the following, we consider \emph{labeled chains}. To this aim, we extend the signature $\set{E}$ by a unary relation symbol $L$ with the meaning that for every node $v$ of a labeled chain~$\C$, we have $v \in L^{\C}$ if the node is labeled and $v \notin L^{\C}$ if it is not. The \emph{root} of $\C$ is defined in the obvious way.

For each $h \in \Npos$, we let $\TT_{2,h}$ be a set of labeled chains of height $h$ which contains for every labeled chain $\C$ of height $h$ precisely one $\C'$ such that $\C\cong\C'$.  Note that, by interpreting labeled chains as encodings of binary numbers, we obtain that $|\TT_{2,2^h}| =  2^{2^h+1}$. 
By $\FF_{2,h}$ we denote the class of all structures over the signature $\set{E,L}$ whose connected components are labeled chains of height $h$. 

The next step is to prove the following lemma. 

\begin{lemma} \label{lem:iso_2-formula}
There is a number $c \in \Npos$ and a sequence $(\iso_h(x, y))_{h \geq 1}$ of $\FO[\set{E,L}]$-formulae of size $\leq c \cdot h$, such
that for each $h \geq 1$, the folllowing holds:
If $\F \in \FF_{2,2^h}$ and $a,b$ are roots of connected components $\C,\C'$ of $\F$, then
\begin{equation*}
  \F \models \iso_h[a,b] \quad\text{if, and only if,}\quad \C \cong \C'.
\end{equation*}
\end{lemma}

\noindent The proof of Lemma \ref{lem:iso_2-formula} is deferred to Appendix~\ref{appendix:lowerbounds}. We are now ready to prove Theorem~\ref{Thm:lowerBound_d_2}.

\begin{proof}[Proof of Theorem \ref{Thm:lowerBound_d_2}] We choose $\sigma \isdef \set{E,L}$. To prove the theorem it suffices to show, that there is a sequence $(\phi_h)_{h\geq 1}$ of $\FO[\sigma]$-sentences of size $\lin{h}$  such that for each $h\geq 1$, every sentence in \BSNF\ that is equivalent to $\phi_h$ on $\FF_{2,2^h}$ has size~$> 2^{2^h+1}$.

 For each $h\geq 1$, let 
 \begin{equation*}
   \phi_h \ \isdef \quad \forall x \, \Bigl(\Root(x) \ \ \impl \ \ \exists y \,\bigl( \Root(y) \ \land \ \lnot x{=}y \ \land \ \iso_h(x,y) \bigr) \Bigr)
 \end{equation*}
 where $\Root(x) \isdef \lnot \exists y \, E(y,x)$ and where $\iso_h(x,y)$ is provided by Lemma~\ref{lem:iso_2-formula}. As a consequence of Lemma~\ref{lem:iso_2-formula}, there is a number $c \in \Npos$ such that for each $h \ge 1$ the size of $\phi_h$ is $\leq c \cdot h$.

For the following, we fix an $h \in \Npos$. Then, for each $\F \in \FF_{2,2^h}$, we have that
\begin{equation*}
     \F \models \phi_h \quad
    \text{if, and only if,} \quad  \text{every labeled chain in $\F$ is isomorphic to another one.}
\end{equation*}
Hence, by the combinatorial Lemma \ref{lem:combinatorial-argument} it follows that every sentence in \BSNF\ that is equivalent to $\phi_h$ on $\FF_{2,2^h}$ and, in particular, on all forests of degree $\leq 2$ over the signature $\sigma$, has size~$> 2^{2^h+1}$.
\end{proof}

\noindent Finally, we will establish the lower bound for classes of structures of degree $\le d$ for~$d \ge 3$.

\begin{theorem}\label{Thm:lowerBound_d_ge2}
  Let $d \ge 3$ be a natural number. There is a signature $\sigma_d$ and a sequence $(\phi_h)_{h\geq 1}$ of $\FO[\sigma_d]$-sentences of size $\lin{h}$
  such that for each $h\geq 1$, every sentence in \BSNF\ that is equivalent to $\phi_h$ on all structures of degree $\leq d$ over the signature $\sigma_d$ has size $> 2^{(d-1)^{2^{h}}}$.
\end{theorem}

Let $d\geq 3$. 
For the proof of Theorem \ref{Thm:lowerBound_d_ge2}, we consider ordered and labeled trees of arity~$d{-}1$, i.e., where every node has at most $d{-}1$ children. Thus, in particular, every $(d{-}1)$-ary tree has degree $\leq d$. 
It is \emph{complete with height $h$} if every path from the root to a leaf has length $h$ and every non-leaf node has exactly $d{-}1$ children.

We represent ordered and labeled trees of arity~$d{-}1$ by structures over the signature $\tau_{d-1} \isdef \set{E_1,\dotsc,E_{d-1},L}$, where the unary predicate $L$ is used as the label (in the same way as for labeled chains) and where the binary predicates $E_1,\dotsc, E_{d-1}$ impose the ordering on the children of every node. I.e., for two nodes $u$ and $v$ of a labeled and ordered tree $\T$, we have $(u,v)\in E^{\T}_i$ if, and only if, $v$ is the $i$-th child of $u$. The \emph{root} of $\T$ is defined in the obvious way.

For each $h\geq 1$, we let $\TT_{d,h}$ be a set of complete labeled and ordered $(d{-}1)$-ary trees with height $h$ which contains for every complete labeled and ordered $(d{-}1)$-ary tree $\T$ with height~$h$ precisely one $\T'$ such that $\T\cong\T'$. 
Note that $|\TT_{d,2^h}| \ge 2^{(d-1)^{2^h}}$.
By $\FF_{d,h}$, we denote the class of all $\tau_{d-1}$-structures whose connected components are complete labeled and ordered $(d{-}1)$-ary trees with height $h$.

The next step is to prove the following lemma.

\begin{lemma}[cf. {\cite[Lemma 21]{FrickGrohe-FO-MSO-revisited}}] \label{lem:iso_d-formula} 
For each natural number $d \ge 2$, there is a number $c_d \in \Npos$ and a sequence  $(\iso_{d-1,h}(x, y))_{h \geq 1}$ of $\FO[\tau_{d-1}]$-formulae of size $\leq c_d \cdot h$, such
that for each $h \geq 1$ the following
holds:  If $\F \in \FF_{d,2^h}$ and $a,b$ are roots of connected components $\T,\T'$ of $\F$, then
\begin{equation*}
  \F \models \iso_{d,h}[a,b] \quad\text{if, and only if,} \quad \T \cong \T'.
\end{equation*}
\end{lemma}

\noindent The proof of Lemma \ref{lem:iso_d-formula} can be found in the appendix.
We conclude with the proof of Theorem \ref{Thm:lowerBound_d_ge2}.

\begin{proof}[Proof of Theorem \ref{Thm:lowerBound_d_ge2}] The proof proceeds in the same manner as the proof of Theorem \ref{Thm:lowerBound_d_2}. For a given $d \ge 3$, we choose $\sigma_d \isdef \tau_{d-1}$. To prove the theorem it suffice to show that there is a sequence $(\phi_h)_{h\geq 1}$ of $\FO[\tau_{d-1}]$-sentences of size $\lin{h}$  such that for each $h\geq 1$, every sentence in \BSNF\ that is equivalent to $\phi_h$ on $\FF_{d,2^h}$ has size $> 2^{(d-1)^{2^{h}}}$.

 For each $h\geq 1$, let 
 \begin{equation*}
   \phi_h \ \isdef \quad \forall x \, \Bigl(   \Root_{d,h}(x) \ \ \impl \ \ \exists y \,\bigl( \Root_{d,h}(y) \ \land \ \lnot x{=}y \ \land \ \iso_{d,h}(x,y) \bigr) \Bigr)
 \end{equation*}
 where $\Root_{d,h}(x) \isdef \lnot \exists y \bigvee\limits_{i=1}^{d-1} E_i(y,x)$ and where the formula  $\iso_{d,h}(x,y)$ is provided by Lemma~\ref{lem:iso_d-formula}. By Lemma~\ref{lem:iso_d-formula}, there is a number $c \in \Npos$ such that for each $h \ge 1$ the size of~$\phi_h$ is $\leq c \cdot h$.

For the following, we fix an $h \in \Npos$. Then, for each $\F \in \FF_{d,2^h}$, we have that
\begin{equation*}
     \F \models \phi_h \quad
    \text{if, and only if,} \quad  \text{every connected component of $\F$ is isomorphic to another one.}
\end{equation*}
By the the combinatorial Lemma \ref{lem:combinatorial-argument} it follows that every sentence in \BSNF\ that is equivalent to $\phi_h$ on $\FF_{d,2^h}$ and, in particular, on all forests of degree $\leq d$ over the signature $\sigma_d$,  has size~$> 2^{(d-1)^{2^h}}$. 
\end{proof}
\noindent Note that the lower bounds proven in this section directly imply that the upper bounds of Section~\ref{section:upper-bounds} are essentially tight.

\section{Conclusion}
\label{section:conclusion}
In this article, we have examined the complexity of constructing the local normal form \BSNF\ from Barthelmann and Schwentick for $\FO$ logic. It turned out that already on the class of all finite forests (of unbounded degree), a non-elementary complexity in terms of the size of the input formula is unavoidable. On the other hand, on classes of structures of bounded degree, \BSNF\ can be computed in elementary time. More precisely, for the class of all structures of degree $\leq d$ for $d = 2$ ($d\geq 3$), our algorithm is $2$-fold ($3$-fold) exponential and worst-case optimal. Our results have corresponding implications for the local Ehrenfeucht-Fra\"iss\'e game and the automata model for $\FO$ logic discussed in \cite{SchB99}. Moreover, the specific shape of \BSNF\ and
 being able to compute \BSNF\ with the same time complexity as Hanf and Gaifman normal form makes it an interesting candidate for new locality-based algorithmic meta-theorems.

\bibliography{bibliography}

\appendix

\section{Upper Bounds}
This part of the appendix provides detailed proofs for the time complexity of the algorithms described in Section~\ref{section:upper-bounds}. 
To simplify notation, we will denote by $S^{\sigma,d}_r(n)$ the class of all functions in
\begin{equation*}
  2^{(n \cdot \nu_d(r))^{\lin{\size{\sigma}}}}.
\end{equation*}
In particular, we will often make use of the fact that $\poly{S^{\sigma,d}_r(n)} \subseteq S^{\sigma,d}_r(n)$.

The following sections are numbered in the same way as the subsections of Section~\ref{section:upper-bounds}.

\subsection{(Positive) Positive Hanf Normal Form}
\label{appendix:upper-bounds:positive-hanf-normal-form}

\begin{proof}[Time complexity of the algorithm described in Lemma~\ref{lem:fo-to-hnfp}]
  We give upper bounds on the time needed to perform the three steps of the algorithm. 
  \begin{enumerate}
  \item According to Theorem~\ref{thm:fo-to-hnf}, the \HNF-formula $\phi^H(\ov{x})$ can be computed in time $S^{\sigma,d}_{4^q}(\size{\phi})$. 
  \item Step~2 can be performed in time $\poly{\size{\phi^H}}$, and thus also in time $S^{\sigma,d}_{4^q}(\size{\phi})$. 
  \item Each negated threshold-counting-formula $\neg\exists^{\geq k}y\,\sph_{\tau}(y)$ with $k\in\Npos$ can be replaced in time
    \lin{k \cdot \size{\phi^H}} and thus, since $k < \size{\phi^H}$, 
    in time $\poly{S^{\sigma,d}_{4^q}(\size{\phi})} \subseteq S^{\sigma,d}_{4^q}(\size{\phi})$.
    
    Consider now a negated type-formulae $\neg\sph_{\rho}(\ov{x}')$, where, for an $r \in \NN$ with $r\leq 4^q$, $\rho$ is a $d$-bounded $r$-type with $n'\in [n]$ centres.
    According to Lemma~\ref{lem:compute-representatives}, the set $\mathfrak{T}^{\sigma,d}_r(n')$ can be computed in time $S^{\sigma,d}_{4^q}(n) \leq S^{\sigma,d}_{4^q}(\size{\phi})$.

    Since the type-formula for each type in $\mathfrak{T}^{\sigma,d}_r(n')$ can be constructed in time $(n'\cdot\nu_d(r))^{\lin{\size{\sigma}}}$ for $n\cdot\nu_d(r)\geq 2$ and in time $\lin{\size{\sigma}}$ otherwise, $\neg\sph_{\rho}(\ov{x}')$ can be replaced in time
    $S^{\sigma,d}_{4^q}(n') \cdot S^{\sigma,d}_{4^q}(\size{\phi}) \leq  S^{\sigma,d}_{4^q}(\size{\phi})$.
    
    As $\phi^H$ contains at most $\size{\phi_H}$ many subformulae of the latter two shapes, we can conclude that Step~\ref{lem:fo-to-hnfp:step3} can altogether be performed in time $\poly{S^{\sigma,d}_{4^q}(\size{\phi})}\subseteq S^{\sigma,d}_{4^q}(\size{\phi})$. 
  \end{enumerate}
  Combining the running times for Step~\ref{lem:fo-to-hnfp:step1}, Step~\ref{lem:fo-to-hnfp:step2}, and Step~\ref{lem:fo-to-hnfp:step3}, we obtain that the algorithm altogether can be carried out in time
  \begin{equation*}
    S^{\sigma,d}_{4^q}(\size{\phi}) \ = \ 2^{(\size{\phi}\cdot\nu_d(4^q))^{\lin{\size{\sigma}}}}.
  \end{equation*}
  This completes the analysis of the time complexity of the algorithm of Lemma~\ref{lem:fo-to-hnfp}.
\end{proof}

\subsection{From Counting-Sentences to BSNF}
\label{appendix:upper-bounds:from-counting-sentences-to-BSNF}

\begin{proof}[Time complexity of the algorithm described in Lemma~\ref{lem:counting-sentences-to-BSNF}] 

Suppose that $\chi$ is of the shape \linebreak$\textsf{Q}y\,\sph_{\tau}(y)$ for $\textsf{Q}$ being either $\exists^{\geq k}$ or $\exists^{=k}$ for some $k\in\Npos$, or $\exists^{=0}$. In each of these cases, $\chi$ has size in $\lin{k^2 + \size{\sph_{\tau}}}$. 

We give upper bounds on the time required by the algorithm to construct the equivalent \BSNF-sentence, distinguishing on the shape of the given counting-sentence $\chi$. 
  \begin{itemize}
  \item If $\textsf{Q} = \exists^{\geq k}$, the described \BSNF-sentence can be constructed in time
    $\lin{ k^2 + \size{\sph_{\tau}}}$ and hence in time $\lin{\size{\chi}}$. 
  \item If $\textsf{Q} = \exists^{=k}$, the described \BSNF-sentence can also be constructed in time
    $\lin{ k^2 + \size{\sph_{\tau}}}$ and thus in time $\lin{\size{\chi}}$.
  \item It is obvious that for $\textsf{Q} = \exists^{=0}$, the described \BSNF-sentence can be constructed in time $\lin{\size{\chi}}$. 
  \end{itemize}
  Thus, in any case, the algorithm runs in time $\lin{\size{\chi}}$ to transform $\chi$ into an equivalent \BSNF-sentence.  
  This completes the analysis of the time complexity of the algorithm of Lemma~\ref{lem:counting-sentences-to-BSNF}.
\end{proof}

\subsection{From Sphere-Formulae to BSNF}
\label{appendix:upper-bounds:from-sphere-formulae-to-BSNF}

\begin{proof}[Time complexity of the algorithm described in Lemma~\ref{lem:type-formulae-to-BSNF}] 
  Observe that the universe of $\A$ consists of at most $n\cdot\nu_d(r)$ elements.
  The first step of the algorithm is to find the $k \in [n]$ connected components $\A_1,\dotsc,\A_k$ of $\A$ and their sets of centres $P_1,\dotsc,P_k$. Starting from the $n$ centres of~$\rho$, this can be done by at most $n\leq n\cdot\nu_d(r)$ runs of a standard graph traversal and thus in time $(n\cdot\nu_d(r))^{\lin{\size{\sigma}}}$ for $n\cdot\nu_d(r)\geq 2$. 

  Afterwards, for each $i \in [k]$, the following steps have to be carried out:
  \begin{enumerate}
  \item Compute the set $\mathfrak{T}_i$. According to Lemma~\ref{lem:compute-representatives}, it takes time in $2^{(\ell_i\cdot\nu_d(r_i))^{\lin{\size{\sigma}}}} = S^{\sigma,d}_{r_i}(\ell_i)$ to compute the set 
    $\mathfrak{T}^{\sigma,d}_{r_i}(\ell_i)$ 
    of representatives of the isomorphism classes of $d$-bounded $r_i$\hbox{-}types with $\ell_i$ centres over the signature $\sigma$.
    For each $r_i$-type $\tau = (B, b_1, \dotsc, b_{\ell_i})$ in~$\mathfrak{T}^{\sigma,d}_{r_i}(\ell_i)$, it takes time polynomial in respect to the size of its universe to decide whether $B \subseteq N_{r_i}^{\B}(b_1)$. Moreover, by Lemma~\ref{lem:types}, it takes time in $ (\ell_i\cdot\nu_d(r))^{\lin{\size{\sigma}}} \cdot 2^{\lin{\size{\sigma}(\ell_i\nu_d(r))^2}}$ and thus time in $S^{\sigma,d}_{r}(\ell_i) \leq S^{\sigma,d}_{r_i}(\ell_i)$ to decide whether $\N_r^{\B}(b_1,\dotsc,b_{\ell_i})\cong\rho_i$. We can conclude that the set $\mathfrak{T}_i$ can be computed in time $\poly{S^{\sigma,d}_{r_i}(\ell_i)}\subseteq S^{\sigma,d}_{r_i}(\ell_i)$. 
  \item For each $\tau \in \mathfrak{T}_i$, construct the formula $\gamma_{i,\tau}(\ov{x})$. 
    For $2 \leq N \isdef \ell_i \cdot \nu_d(r_i)$, this takes time in
    \begin{equation*}
      \lin{ N + N \cdot \size{\sigma}\cdot \log r_i + n \cdot \size{\sigma}\cdot \log r_i + n^2 } + N^{\lin{\size{\sigma}}}
      \quad \leq \quad N^{\lin{\size{\sigma}}}.
    \end{equation*}
    For this upper bound, recall that $\ell_i\leq n$ and that $2r{+}1 \leq r_i$. 
  \item Construct the formula $\gamma_i(\ov{x})$.
    Using the upper bound on the time needed to construct the formulae $\gamma_{i,\tau}$ for each $\tau\in\mathfrak{T}_i$ and the upper bound on the size of the set $\mathfrak{T}_i$, this takes time in
    ${(\ell_i\cdot\nu_d(r_i))}^{\lin{\size{\sigma}}} \cdot S_{r_i}^{\sigma,d}(\ell_i)$ for $\ell_i\cdot \nu_d(r_i)\geq 2$.
  \end{enumerate}
  After these preparations, the construction of the final \BSNF-formula $\alpha^B(\ov{x})$ takes time in
  \begin{equation*}
    S^{\sigma,d}_{n(2r+1)}(n) \ = \ 
    2^{(n\cdot\nu_d(n(2r{+}1)))^{\lin{\size{\sigma}}}}
    \ \subseteq \ 
    2^{\nu_d(n(2r{+}1))^{\lin{\size{\sigma}}}}
  \end{equation*}
  as $k \leq n$, $\ell_i \leq n$ and $r_i \leq n(2r{+}1)$ for each $i\in [k]$. 
  This completes the analysis of the time complexity of the algorithm of Lemma~\ref{lem:type-formulae-to-BSNF}.
\end{proof}

\subsection{Proof of Theorem~\ref{thm:bsnf-upper-bound}}
\label{appendix:upper-bounds:proof-of-main-theorem}

\begin{proof}[Time complexity of the algorithm described in Theorem~\ref{thm:bsnf-upper-bound}] 
  We analyse the time complexity of the steps of the algorithm.
  \begin{enumerate}
  \item The algorithm of Lemma~\ref{lem:fo-to-hnfp} computes the $\HNF_+$-formula $\phi^H_+(\ov{x})$ in time $S^{\sigma,d}_{4^q}(\size{\phi})$ and with locality radius at most $r \isdef 4^q$. 
  \item The algorithm of Lemma~\ref{lem:counting-sentences-to-BSNF} requires time linear in the size of the counting-sentence for each counting-sentence of $\phi^H_+(\ov{x})$. Thus, the counting-sentences in $\phi^H_+(\ov{x})$ can be replaced by equivalent \BSNF-sentences in time 
    $(S^{\sigma,d}_r(\size{\phi}))^2 \subseteq S^{\sigma,d}_{r}(\size{\phi})$. 
    \item The algorithm of Lemma~\ref{lem:type-formulae-to-BSNF} requires time in $S^{\sigma,d}_{n(2r{+}1)}(1)$ for each type-formula of $\phi^H_+(\ov{x})$. Thus, the type-formulae in $\phi^H_+(\ov{x})$ can be replaced by equivalent \BSNF-formulae in time
      $ S^{\sigma,d}_r(\size{\phi}) \cdot S^{\sigma,d}_{n(2r+1)}(1) \leq S^{\sigma,d}_{n(2r+1)}(\size{\phi})$.
    \item
      Observe that the positive Boolean combination $\psi(\ov{x})$ of \BSNF-formulae, obtained at the end of Step~(3), has size in $S^{\sigma,d}_{(n+1)(2r+1)}(\size{\phi})$.       The construction of the final \BSNF-formula is described by an induction over the shape of $\psi(\ov{x})$.
      It can be performed efficiently in three steps:
      \begin{enumerate}
      \item Rename the variables $y_1,\dotsc,y_m,z$ in the prefixes of each \BSNF-formula of the shape $\exists y_1 \cdots y_m \forall z \phi(\ov{x},y_1,\dotsc,y_m,z)$ such that no existentially quantified variable occurs in more than one of the \BSNF-formulae of $\psi(\ov{x})$ and such that all \BSNF-formulae of $\psi(\ov{x})$ use the same universally quantified variable $z$. 
      \item Use the list of renamed variables, obtained in Step~(a), to construct the prefix of the final \BSNF-formula according to the shape of the Boolean combination of $\psi(\ov{x})$. 
      \item Use the variable renamings, obtained in Step~(a), to construct the local formula of the final \BSNF-formula according to the shape of the Boolean combination of $\psi(\ov{x})$. 
      \end{enumerate}
      Each of these steps can be carried out in time polynomial in the size of $\psi(\ov{x})$, and thus in time $S^{\sigma,d}_{(n+1)(2r+1)}(\size{\phi})$. 
  \end{enumerate}
  Adding up the time for Step~(1), Step~(2), Step~(3), and Step~(4), the $d$-equivalent \BSNF-formula for the input formula $\phi(\ov{x})$ can be computed in time
  \begin{equation*}
    S^{\sigma,d}_{(n+1)(2r+1)}(\size{\phi}) \ = \ 2^{(\size{\phi}\cdot\nu_d((n+1)(2\cdot 4^q+1)))^{\lin{\size{\sigma}}}}.
  \end{equation*}
  This completes the analysis of the time complexity of the algorithm of Theorem~\ref{thm:bsnf-upper-bound}.
\end{proof}

\section{Lower Bounds}

\subsection{Proofs of Lemma~\ref{lem:iso_2-formula} and Lemma~\ref{lem:iso_d-formula}}
\label{appendix:lowerbounds}
\begin{proof}[Proof of Lemma~\ref{lem:iso_2-formula}]
In fact, Lemma~\ref{lem:iso_2-formula} is a special case of Lemma~\ref{lem:iso_d-formula}.
I.e., for each~$h\geq 1$, $\iso_h(x,y)$ can be chosen as the formula $\iso_{2,h}(x,y)$ provided by Lemma~\ref{lem:iso_d-formula}. The only modification necessary is to replace the relation symbol $E_1$ from the signature $\tau_1 = \set{ E_1, L}$ by the relation symbol $E$ from the signature $\set{E}$.
\end{proof}

In the following, $\FF_d$ denotes the class of all labeled and ordered trees of arity $d{-}1$. 
For the proof of Lemma~\ref{lem:iso_d-formula} we need the following lemma,
which provides, for distances $\ell\geq 0$, formulae of size logarithmically in $\ell$ that recognise co-reachable pairs of nodes up to distance~$\ell$ in forests from $\FF_d$. 
In a forest $\F\in\FF_d$, we call two nodes $b,b' \in F$  \emph{co-reachable with distance~$\le \ell$ from $a$ and $a'$} if, and only if, the following holds: 
\begin{quote}
There is a number $n \le \ell$ and sequences $a = c_0 , \ldots , c_n = b$ and $a'= c_0' , \ldots , c_n' = b$ of nodes in $\F$, such that for every $i \in [0,n]$, there is a $j \in [d-1]$, such that $(c_i,c_{i+1})$ and $(c_i',c_{i+1}')$ belong to the $j$-th successor relation $E_j$ of $\F$.
\end{quote}
 Intuitively, this means that $b$ and $b'$ can be reached from $a$ and $a'$, respectively, by paths using the same successor relations in the same order.

 In the context of binary trees, this lemma was formulated as Lemma 25 in \cite{FrickGrohe-FO-MSO-revisited}.
 
\begin{lemma}[cf. {\cite[Lemma 25]{FrickGrohe-FO-MSO-revisited}}] \label{lem:coreach}
Let $d \ge 2$. For each $\ell \ge 0$, there is an $\FO[\tau_{d-1}]$-formula $\coreach_{d,\ell}(x, y, x', y' )$ such that for every $\F \in \FF_d$ and all $a, b, a' , b' \in \F$, we have
\begin{align*}
	& F \models \coreach_{d,\ell}[a, b, a' , b' ] \\
\text{if, and only if, }\quad & b, b' \text{ are co-reachable from $a, a'$ with distance } \le \ell. \end{align*}
The size of $\coreach_{d,\ell}(x, y, x', y' )$ grows logarithmic in $\ell$, i.e., it is in $\lin{d \cdot \log \ell}$ for $\ell \geq 2$.
\end{lemma}

\begin{proof} The construction of the formula proceeds by an induction over $\ell$ and is adapted from the proof of \cite[Lemma 25]{FrickGrohe-FO-MSO-revisited}. For $\ell=0$, we let \[ \coreach_{d,0}(x,y,x',y') \ \isdef \  x=y \, \land \, x'=y' \] and for $\ell=1$, we define
  \[ \coreach_{d,1}(x,y,x',y') \ \isdef \quad  \coreach_{d,0}(x,y,x',y') \ \lor \ \bigvee\limits_{i=1}^{d-1} \bigl(E_i(x,y) \ \land \ E_i(x',y') \bigr) \]
  Then, for all $\ell \ge 1$, we let 
  \begin{align*}
    \coreach_{d,2\ell}(x,y,x',y') \ \isdef \quad   \exists z \exists z' \forall u \forall v \forall u' \forall v' \Big( \big( 
    ( u=x \, \land \, u'=x' \, \land \, v=z \, \land \, v'=z' 	) \\
    \lor
    ( u=z \, \land \, u'=z' \, \land \, v=y \, \land \, v'=y')
    \big)\\ \impl \coreach_{d,\ell}(u,v,u',v')
    \Big)
  \end{align*}
  and 
  \begin{align*}
    \coreach_{d,2\ell+1}(x,y,x',y') \ \isdef \quad  \exists z \exists z'\bigl( \coreach_{d,1}(x,z,x',z') \, \land \, \coreach_{d,2 \ell}(z,y,z',y') \bigr)
  \end{align*}
  The upper bound on the size of the formula follows from a straightforward induction on~$\ell$. 
\end{proof}

\begin{proof}[Proof of Lemma~\ref{lem:iso_d-formula}] 
Let $d \ge 2$ and recall that $\tau_{d-1}$ denotes the signature $\set{E_1,\ldots, E_{d-1}, L}$.

The task is to define, for each $h\geq 1$, a $\FO[\tau_{d-1}]$-formula $\iso_{d,h}(x, y)$, such that for every $\F \in \FF_{d,2^h}$ 
 and all roots $a,b$ of connected components $\T,\T'$ of $\F$, we have: \[  \F \models \iso_{d,h}[a,b] \quad\text{if, and only if,} \quad \T \cong \T'. \] 
 
\noindent Using the formula $\coreach_{d,\ell}$, defined in Lemma~\ref{lem:coreach}
, we let
\[ \iso_{d,h}(x,y) \ \isdef \ \forall x' \, \forall y' \Big(\coreach_{d,2^h}(x,x',y,y') \impl \big( L(x')\ \gdw \ L(y') \big)\Big).\]
It is easy to verify, that the formula ensures that all nodes $x'$ and $y'$ which are on corresponding positions in the trees below $x$ and $y$, respectively, are labeled in the same way. Together with the fact, that both trees below $x$ and $y$ are complete ordered and labeled $(d{-}1)$-trees with height $2^h$, it follows that the formula $\iso_{d,h}[a,b]$ ensures for two root nodes $a$ and $b$ of connected components $\T,\T'$ in a forest $\F$ from $\FF_{d,2^h}$ that $\T \cong \T'$.

Finally, by induction on the recursive definition of the formula $\coreach_{d,\ell}$ it is straightforward to verify, that there is a number $c_d \in \Npos$ such that for all $h \ge 1$, the formula $\iso_{d,h}$ has size $\le c_d \cdot h$. This completes the proof of Lemma~\ref{lem:iso_d-formula}.
\end{proof}

\end{document}